\documentclass{article}
\usepackage{ijcai20}

\usepackage{times}

\usepackage{soul}
\usepackage{url}
\usepackage[hidelinks]{hyperref}
\usepackage[utf8]{inputenc}
\usepackage[small]{caption}
\usepackage{graphicx}
\usepackage{amsmath}
\usepackage{amsthm}

\usepackage{amssymb}
\usepackage{amsfonts}

\usepackage{color}

\usepackage{booktabs}
\usepackage{algorithm}
\usepackage{algorithmic}
\newtheorem{theorem}{Theorem}
\newtheorem{definition}{Definition}
\newtheorem{lemma}{Lemma}
\newtheorem{mechanism}{Mechanism}
\newtheorem{proposition}{Proposition}

\title{Sybil-proof Answer Querying Mechanism}

\author{
Yao Zhang%$^1$
\and
Xiuzhen Zhang%$^1$
\And
Dengji Zhao%$^{1}$
\affiliations
Shanghai Engineering Research Center of Intelligent Vision and Imaging, ShanghaiTech University%\\
\emails
\{zhangyao1, zhangxzh1, zhaodj\}@shanghaitech.edu.cn%,
}
\begin{document}

\maketitle

\begin{abstract}
We study a question answering problem on a social network, where a requester is seeking an answer from the agents on the network. The goal is to design reward mechanisms to incentivize the agents to propagate the requester's query to their neighbours if they don't have the answer. Existing mechanisms are vulnerable to Sybil-attacks, i.e., an agent may get more reward by creating fake identities. Hence, we combat this problem by first proving some impossibility results to resolve Sybil-attacks and then characterizing a class of mechanisms which satisfy Sybil-proofness (prevents Sybil-attacks) as well as other desirable properties. Except for Sybil-proofness, we also consider cost minimization for the requester and agents' collusions.
%In the problem of crowdsourcing, a requester can utilize a social network to seek for an answer by incentivizing agents to propagate her query. Such an incentive query mechanism has been well studied and has also been applied in some real-world tasks. However, these mechanisms are always vulnerable to Sybil-attacks, i.e., some agents can get more reward by creating fake identities. Therefore, to tackle this problem, many studies have put great effort on Sybil-proofness mechanisms in different settings. In this paper, we advance the state of the art in Sybil-proof answer querying mechanisms. We characterize a class of reward mechanisms named Double Geometric Mechanisms (DGM) which satisfies Sybil-proofness under a dominant strategy implementation. We also show the performance of DGM in the cost minimization and the approximation of collusion-proofness.
\end{abstract}

\section{Introduction}
The development of online social networks has offered many opportunities for people to collaborate remotely in real time, such as 
% stimulated many task-involved agents to query help or resources via various crowdsourcing systems including 
P2P file-sharing network (e.g., BitTorrent) and Q\&A platforms (e.g., Quora and Stack Overflow). Inspired by these applications, there are rich theoretical studies to look at the mechanism design problems on social networks~\cite{rahman2009survey,emek2011mechanisms,li2017mechanism}.

In this paper, we focus on the answer/resource querying mechanisms via a social network where a requester is searching for a single answer from the network. Here we have two main challenges. The first is that the requester is only connected to a few agents (her neighbours) on the network and she needs to find a way to inform the other agents on the network if her neighbours do not have the answer. This can be achieved by incentivizing her neighbours to propagate the query to their neighbours (diffusion incentives). 
%, i.e., how to incentivize agents to propagate the query to more individuals if they do not have the answer. 
\citeauthor{kleinberg2005query}~\shortcite{kleinberg2005query} first formulated diffusion incentives in query networks. Later on, \citeauthor{arcaute2007threshold}~\shortcite{arcaute2007threshold} and \citeauthor{kota2010threshold}~\shortcite{kota2010threshold} studied the threshold of rewards needed to incentivize the query network. \citeauthor{kleinberg2007cascading}~\shortcite{kleinberg2007cascading} further discussed whether agents will act as the requester expects under different reward settings.
Similar approach has been applied in advertising~\cite{li2012diffusion}, auctions~\cite{zhao2018selling}, recommendations~\cite{margaris2016recommendation} and others. 

Once the first challenge is sovled, then an immediate new challenge is Sybil attacks, where an agent pretends to be multiple agents to gain more rewards~\cite{conitzer2010false}. 
%Another challenge is the resistance to the agents' manipulations, which may damage the interests of the requester. One major concern of manipulations is Sybil attacks~\cite{mohaisen2013sybil,conitzer2010false}, where agents pretend to be more agents. 
Sybil attack has been studied in many other applications such as multi-level marketing ~\cite{emek2011mechanisms,drucker2012simpler,shen2019multi}, social choice~\cite{conitzer2010using} and blockchains~\cite{babaioff2012bitcoin,ersoy2018transaction}. 
%However, in query networks, this has not been investigated.
%When it turns to query networks, although many mechanisms work well on incentivizing the propagation (e.g. the success of the winning team in DARPA 2009 red balloon challenge~\cite{pickard2011time}), they are vulnerable to Sybil attacks. 
For query networks, Sybil attack has been investigated under various settings. For example, \citeauthor{seuken2014sybil}~\shortcite{seuken2014sybil} studied the trade-off between the transitive trust and Sybil attacks in P2P file-sharing networks. \citeauthor{chen2013sybil}~\shortcite{chen2013sybil} proposed a mechanism based on the query incentive network~\cite{kleinberg2005query} which achieves Sybil-proofness in expectation. \citeauthor{nath2012mechanism}~\shortcite{nath2012mechanism} identified a collusion-proof mechanism with approximated Sybil-proofness. However, all the existing work has only discussed Sybil-proofness in approximation or expectation on query networks.

Therefore, in this paper, we aim to solve both challenges in a dominant strategy implementation for query networks. 
%Therefore, Sybil-proof mechanisms in query incentive networks under a dominant strategy implementation still need to be further investigated. 
We first demonstrate the difficulties to solve the challenges by proving the impossibility results. We then characterize a class of mechanisms, called double geometric mechanism (DGM), to satisfy the desirable properties. Except for the diffusion incentive and Sybil-proofness, we also look at the cost minimization problem for the requester and preventing agents' collusion. We will show that Sybil-proofness and collusion-proofness are not compatible.
%In this paper, we propose a double geometric mechanism (DGM) which is Sybil-proof under the deterministic query networks setting. Moreover, we also introduce many desired properties including incentive compatible, individual rational, budget constrained and other important requirements. We show that all these desired properties fully characterize the DGM. Furthermore, we discuss the cost minimization issues, approximation of collusion-proofness in query networks.

%The remainder of the paper is organized as follows. Section~\ref{sec:rw} summaries the most related works and highlights our contributions. In Section~\ref{sec:model}, we describe the model of the query network and introduce the desired properties of a centralized reward mechanism. Then we state several impossible results in Section~\ref{sec:imp}. Next, we propose DGM, characterize it with our desired properties and discuss several extensions in Section~\ref{sec:dgm}. At last, we conclude our results and future directions in Section~\ref{sec:con}.

\subsection{Related Work}\label{sec:rw}
%In this section, we briefly illustrate some of the most related work to ours and show the most significant difference and extension in our work.
The query incentive network model was first proposed by  \citeauthor{kleinberg2005query}~\shortcite{kleinberg2005query}, where each agent in a $d$-ary tree network has the same probability to hold the answer and the actual query network is generated by a branching process (the cost in the query process is negligible). They considered a decentralized reward mechanism where each agent strategically chooses a fixed amount of reward to offer to her children if she can receive the answer from them. Different from the fixed reward mechanism, \citeauthor{cebrian2012finding}~\shortcite{cebrian2012finding} proposed a split contract mechanism in the same setting, which was motivated by the success of the winning strategy in the DARPA 2009 Network Challenge~\cite{pickard2011time}. In their split contract mechanism, each agent should determine the splits of the reward received from her parent to offer to her children. In both studies, they only considered the Nash equilibrium implementation, while we consider dominant strategy implementation here.

Moreover, the studies mentioned above mainly focused on propagating the query in the network and they did not consider agents' Sybil attacks. However, when we consider centralized query mechanisms where the reward distribution is decided by the requester, Sybil attack is a problem if the requester cannot verify their indenties~\cite{douceur2002sybil}. %which can prevent a single agent from getting more reward than their contributions by pretending to be multiples agents. 
In the previous work, \citeauthor{chen2013sybil}~\shortcite{chen2013sybil} proved that split contract mechanisms cannot prevent Sybil attacks. They proposed direct referral mechanism to deal with this problem by allocating the majority of the reward to the winner (the agent who holds the answer) as well as her parent, i.e., the direct referral. %, while other agents on the winning path (the path from the requester to the winner) will get the minimum reward. 
However, the Sybil-proofness of the direct referral mechanism is only in expectation, which means that Sybil attack might be beneficial for an agent if the agent can acquire more knowledge about the network.

Again for Syblil-proofness, in the setting of dominant strategy implementation, \citeauthor{nath2012mechanism}~\shortcite{nath2012mechanism} studied the split contract mechanism design. They identified a set of desirable properties including Sybil-proofness and collusion-proofness and proved that no mechanism can satisfy them simultaneously under some conditions. They also examined the well-known geometric mechanism and showed that it satisfies collusion-proofness and approximated Sybil-proofness. 

Different from the above, in this paper, we investigate centralized Sybil-proof reward mechanism design in the query network under dominant strategy implementation. We characterize a class of mechanisms to achieve Sybil-proofness and other properties in the query network.

The remainder of the paper is organized as follows. 
%Section~\ref{sec:rw} summaries the most related works and highlights our contributions. 
Section~\ref{sec:model} describes the model of the query network and introduces the desirable properties of the reward mechanism. Then we prove the impossibility results in Section~\ref{sec:imp}. Following that, we propose our double geometric mechanism and characterize its uniqueness under different properties in Section~\ref{sec:dgm}. Finally, we conclude and discuss the future work in Section~\ref{sec:con}.

\section{The Model}\label{sec:model}
We consider a question answering setting where a requester $r$ is seeking the answer of a question from a set of agents. The agents are connected via their social connections such as friendship and $r$ connects to a subset of them. $r$ will first ask her neighbours for the answer and if her neighbours do not have the answer, $r$ wants her neighbours to further propagate the question to their neighbours and so on. We assume that there is at least one agent who holds the answer and the answer is unique/verifiable. The goal of $r$ is to design a reward mechanism to find the answer from the network. Ideally, we want each agent to offer the answer if she has, otherwise, to propagate the question to her neighbours if there is any. Formally, the propagation process will build a query tree $T = (V,E)$ rooted at $r$, where $V$ is the set of all agents, including the requester $r$, who have been asked for the answer, and each edge $e = (i,j)\in E$ means that agent $i$ has propagated the query to agent $j$ and $j$ has either offered the answer or propagated the question to her neighbours. For each agent $i$ in $T$, let $p(i)$ be $i$'s direct parent and $s(i)$ be $i$'s direct children set.

%Consider a requester $r$ who is willing to acquire one answer. The requester first queries the answer to her friends and if they cannot handle the deal, they can propagate the query to their friends. The propagation continues until an agent, who receives and accepts the invitation, has the capability to provide the answer and is willing to give her answer to the requester. If such an agent does not exist, the query is failed and the requester does not need to give reward to other agents. Hence, we assume that there is at least one agent who holds the answer in the network. More specifically, we assume the resulting query network is a tree $T = (V,E)$ rooted at $r$, where $V$ is the set of all agents in the resulting query network and each $e = (i,j)\in E$ means that agent $i$ has propagated the query to agent $j$ and $j$ has accepted the invitation. For an agent $i$, we call the agent with an edge to agent $i$ as $i$'s parent, denoted by $p(i)$ and the set of agents with an edge from agent $i$ as $i$'s children, denoted by $s(i)$.

We assume that there is one agent in $T$ who offered the answer, which is called the winner, denoted by $w$. It is clear that $w\neq r$ and if there are multiple agents offered the answer, we choose the one with the smallest depth with random tie-breaking. We call the path from $r$ to $w$ a winning path, denoted by $p_w = (i_1, i_2, \dots, i_n)$, where $i_1 \in s(r)$, $i_n = w$, $(i_j, i_{j+1})\in E$ for all $1\leq j< n$ and $n$ is the length of the winning path.

%Note that there are two special agents in $V$, the first one is the requester $r$, which is referred as the root node (no edges go into $r$ and every other node is reachable from $r$); the other one is the node who provides the answer, which is referred as the winner node $w$. Note that $w\neq r$, otherwise she does not need to query answers from her friends. Moreover, if there are multiple agents willing to provide the answer, we will choose the one with the smallest depth. We call the path from $r$ to $w$ a winning path, i.e., the winning path $p_w$ is a sequence of agents $(i_1, i_2, \dots, i_n)$, where $i_1 \in s(r)$, $i_n = w$, $(i_j, i_{j+1})\in E$ for any $1\leq j< n$ and $n$ is the length of the winning path.

Given the above setting, the requester needs to design a reward mechanism $M: \mathcal{T} \rightarrow \mathbb{R}^{V\setminus\{r\}}$ to incentivize the agents to find the answer, where $\mathcal{T}$ is the space of all possible resulting query trees and the output is the reward allocation for each player in the tree. Denote the reward allocated to agent $i$ by $x_i$. In this paper, we focus on path mechanisms, which are mostly studied in the literature~\cite{kleinberg2005query,chen2013sybil}.

\begin{definition}
A reward mechanism is called a \textbf{path mechanism} if
\begin{enumerate}
 \item it only assigns non-zero rewards to the agents on the winning path, i.e., $x_i = 0$ for all $i\notin p_w$,
 \item the reward distributed to an agent on the winning path only depends on her depth and the length of the path. 
\end{enumerate} 
That is, a path mechanism $M$ can be represented by a \textbf{reward function} $x:\mathbb{N}\times \mathbb{N} \rightarrow \mathbb{R}$ for the agents on the winning path, where the first parameter is the agent's depth and the second parameter is the length of the winning path. % i.e., for a winning path $p_w = (i_1, i_2, \dots, i_n)$, $x_{i_j} = x(j, n)$ for any $1\leq j\leq n$, where $j$ is the position on the winning path and $n$ is the length of the winning path.
\end{definition}

A path mechanism only rewards the agents on the winning path in the resulting query network. This will not weaken our results since the agents on the winning path made the actual contribution for seeking the answer. We assume that the cost for the query propagation is negligible as the propagation is often easy or automated~\cite{yu2003searching}. Therefore, the literature has also focused on path mechanisms.%~\cite{kleinberg2005query,chen2013sybil}.%On the other hand, we will see that path mechanisms can be easily generalized for general graph query networks in Section~\ref{sec:con}.

In the following, we define the desirable properties for path mechanisms. Firstly, it should incentivize agents to offer the answer or further propagate the query to all their neighbours; otherwise, the query will stop at the requester's neighbours. We call this property incentive compatibility. % we want to make sure that every agent will propagate the query or provide the answer as the requester expected.

\begin{definition}
A path mechanism is \textbf{incentive compatible (IC)} if for all agent $i$,
$x_i \geq x_i'$, where 
\begin{itemize}
\item $x_i$ is the reward $i$ receives if $i$ truthfully reports her answer, if $i$ has the answer, or propagates the query to all her neighbours if $i$ does not have the answer,
\item $x_i^\prime$ is the reward $i$ receives if $i$ behaves differently. 
\end{itemize}

%(1) for any agent $i\in p_w$ other than $w$, $x_i \geq x_i'$, where $x_i'$ is the reward of $i$ in the resulting query network $T'$ where $i$ has less children; (2) for the agent $w$, $x_w \geq x_w'$, where $x_w'$ is the reward of $w$ in the resulting query network $T'$ where $w$ does not provide the answer and propagates the query.
\end{definition}

An incentive compatible mechanism guarantees that it will always find the answer if there is one in the network.
%Intuitively, the property of incentive compatibility ensures that every agent participated in will provide the answer if she has one and is willing to propagate the query to all her children if she has no answer. 
Secondly, we also require that every agent is not forced to participate in the mechanism, i.e., their reward should not be negative, which is called individual rationality.

\begin{definition}
A path mechanism is \textbf{individually rational (IR)} if for all agent $i\in p_w$, $x_i \geq 0$, i.e., $x(j,n)\geq 0$ for all $j,n\in \mathbb{Z}^+$, $j\leq n$. It is \textbf{strongly individually rational (SIR)} if for all agent $i\in p_w$, $x_i > 0$, i.e., $x(j,n)> 0$ for all $j,n\in \mathbb{Z}^+$, $j\leq n$.
\end{definition}

Notice that IR property can be easily satisfied even if the mechanism does not give any reward to any agent. Therefore, we also look for strongly IR path mechanisms, which at least reward something for every agent in the winning path. 
%Intuitively, the property of individual rationality ensures that every agent on the winning path can get a non-negative reward. Besides, strongly individual rationality promises every agent on the winning path a strictly positive reward. The property of SIR is mainly focused on in our work since in many real-world applications, the requester may strongly rely on these agents to get the answer. 
At the same time, the requester may also want to control the total reward distributed. % that the total rewards distributed is bounded since the requester usually has a finite budget to query the answer.

\begin{definition}
A path mechanism is \textbf{budget constrained (BC)} if there exists a constant $B<\infty$ such that
\[ \sum_{i\in p_w} x_i = \sum_{j=1}^n x(j,n) \leq B \]
for all resulting query network $T$, for all winning path $p_w$ of length $n$ in $T$.
\end{definition}

%Intuitively, if the requester has a maximum budget of $B$, the property of budget constraint ensures that the cost of the requester will not exceed the budget. 
Next, we consider the most important property of Sybil-proofness, which requires that the mechanism should be resistant to Sybil-attacks. Since the resulting query network is a tree and we focus on path mechanisms, fake identities which are not on the winning path benefit nothing. So there is only one kind of meaningful Sybil-attacks for agent $i$, i.e., pretending to be multiple agents from agent $p(i)$ to agents in $s(i)$ (extending the paths from $p(i)$ to $s(i)$).

\begin{definition}
A path mechanism is \textbf{Sybil-proof (SP)} if 
for any winning path $p_w = (i_1, i_2, \dots, i_n)$, if an agent $i_j\in p_w$ extends $p_w$ by making $m$ copies of herself to get the new winning path  $p_w^{\prime} = (i_1, \dots, i_{j-1}, i_j^0, i_j^1, \dots, i_j^m, i_{j+1}, \dots, i_n)$, we have:
\begin{equation}\label{eq:sp}
x_{i_j} = x(j,n) \geq \sum_{k=0}^{m} x'_{i_j^k} =\sum_{k=0}^{m} x(j+k, n+m)
\end{equation}
%for any two winning paths $p_w = (i_1, i_2, \dots, i_n)$ and $p_w' = (i_1', i_2', \dots, i_{n'}')$ of lengths $n$ and $n'$ respectively where $n'>n$, we have
%\[ x_{i_j} \geq \sum_{k=j}^{j'} x_{i_k'} \]
%for any $1\leq j\leq n$ and $j<j'\leq n'$. Or alternatively,
%\begin{equation}\label{eq:sp}
%    x(j,n) \geq \sum_{k=0}^{m} x(j+k, n+m)
%\end{equation}
%for any $j,n,m\in \mathbb{Z}^+$, $j\leq n$.
\end{definition}

Intuitively, if an agent pretends to be multiple agents along the winning path, the property of Sybil-proofness ensures that the total rewards she can get from the multiple agents is not more than what she will get originally. % knows she is on the winning path. 
In our setting, Sybil-proofness can be easily verified as stated in Proposition~\ref{prop:sp}. % it can be seen that the resistance to an agent's splitting herself into two agents is sufficient to maintain the property of Sybil-proofness.

\begin{proposition}
\label{prop:sp}
A path mechanism is Sybil-proof if and only if
\begin{equation*}
    x(j,n) \geq x(j,n+1) + x(j+1,n+1)
\end{equation*}
for all $j,n\in \mathbb{Z}^+$, $j\leq n$.
\end{proposition}

\begin{proof}
(``$\Rightarrow$") If a path mechanism is Sybil-proof, then in Inequality~(\ref{eq:sp}), let $m=1$, it can be easily derived that $x(j,n) \geq x(j,n+1) + x(j+1,n+1)$ for all $j,n\in \mathbb{Z}^+$, $j\leq n$.

(``$\Leftarrow$") If a path mechanism satisfies $x(j,n) \geq x(j,n+1) + x(j+1,n+1)$ for all $j,n\in \mathbb{Z}^+$, $j\leq n$, then
\begin{align*}
    x(j, n) & \geq x(j,n+1) + x(j+1,n+1) \\
    & \geq (x(j,n+2) + x(j+1,n+2)) \\
    & + (x(j+1,n+2) + x(j+2,n+2)) \\
    & \geq \sum_{k=0}^m \binom{m}{k} x(j+k,n+m) \\
    & \geq \sum_{k=0}^m x(j+k,n+m)
\end{align*}
for all $j,n,m\in \mathbb{Z}^+$, $j\leq n$. Hence, the mechanism is Sybil-proof.
\end{proof}

Sybil-proofness says that an agent cannot gain by making multiple copies of herself, while another property says multiple agents could also not collude together to receive a better reward, which is called collusion-proofness~\cite{chen2018collusion}. We will show that these two properties cannot be satisfied together in general.

%On the other hand, collusion-proofness is another critical issue that has been considered in the literature~\cite{chen2018collusion}. This property requires that a group of agents cannot get more reward if they pretend to be a single agent.

\begin{definition}
A path mechanism is \textbf{$\lambda$-collusion proof ($\lambda$-CP)} $(\lambda\in \mathbb{Z}^+ \text{ and } \lambda>1)$ if
\begin{equation}\label{eq:gre}
    x(j,n) \leq \sum_{k=0}^{\lambda'-1} x(j+k,n+\lambda'-1)
\end{equation}
for all $1 \leq \lambda'\leq \lambda$, all $j,n\in \mathbb{Z}^+$ such that $j\leq n-\lambda' +1$. If $\lambda$-collusion proof holds for all $\lambda >1$, i.e.,
\[ x(j,n) \leq \sum_{k=0}^m x(j+k,n+m) \]
for all $j,n,m\in \mathbb{Z}^+$ such that $j\leq n$, then, we call this mechanism \textbf{collusion-proof (CP)}.
\end{definition}

Intuitively, $\lambda$-collusion proofness indicates that any group of agents with size less than $\lambda$ cannot get more reward if they pretend to be a single agent. Collusion proofness requires this to be held for all group sizes. 
%prevents a group of any size from cheating. 
In real-world applications, agents' ability to form a collusion is always limited as, for example, it is not easy to collude for agents far from each other in the network. Thus, $\lambda$-collusion proofness is an applicable approximation of collusion proofness in practice.
%Intuitively, if a group of $\lambda$ (or less) agents on the winning path can be replaced by a single agent and vice versa, the total rewards distributed to the group and to the single agent should be equal. Lemma~\ref{lemma:group} shows that $\lambda$-group reward equivalence is a strong property that may not be satisfied with the other desired properties.

Lastly, to guarantee the reward distributed to agents on the winning path is sufficient, we require that each agent on the winning path can get at least a certain fraction of the winner's reward. This property is inspired by the efficiency of split contracts~\cite{cebrian2012finding}.

\begin{definition}
A path mechanism is \textbf{$\rho$-split secure ($\rho$-SS)}, $0<\rho<1$, if for all agent $i_j\in p_w\setminus \{w\}$, $x_{i_j} \geq \rho x_{i_{j+1}}$, i.e., $x(j,n) \geq \rho x(j+1,n)$ for all $j,n\in \mathbb{Z}^+$, $j< n$.
\end{definition}

Intuitively, the property of $\rho$-split security ensures that each agent on the winning path other than the winner has at least $\rho$ fraction of the reward distributed to her children on the winning path. This guarantee will encourage agents to propagate the query in reality.

%In the next section, we will design a reward mechanism that satisfies these desired properties.

\section{Impossibility Results}\label{sec:imp}
Before characterizing the reward mechanisms that satisfy the desirable properties, we prove several impossibility results in this section.

The first impossibility is that if strongly IR is required (i.e., each agent on the winning path is rewarded more than zero), Sybil-proofness and collusion-proofness cannot be held together.

\begin{lemma}\label{lemma:group}
An SIR path mechanism cannot be both Sybil-proof and $\lambda$-collusion-proof for all $\lambda \geq 3$.
\end{lemma}

\begin{proof}
If such a mechanism exists, then according to the inequality in Proposition~\ref{prop:sp}, let $m = \lambda -1$ and we have
\begin{align}\label{eq:lemma1}
    x(j, n) & \geq \sum_{k=0}^{\lambda-1} \binom{\lambda-1}{k} x(j+k,n+\lambda-1)\notag \\
    & \geq \sum_{k=0}^{\lambda-1} x(j+k,n+\lambda-1)
\end{align}
for any $j,n\in \mathbb{Z}^+$, $j\leq n-\lambda +1$. Together with Inequality~\eqref{eq:gre}, we know that
\[ x(j,n) = \sum_{k=0}^{\lambda-1} x(j+k,n+\lambda-1) \]

Hence, the middle part of Inequality~\eqref{eq:lemma1} is equal to both the left hand side and the right hand side.
\[  \sum_{k=0}^{\lambda-1} \binom{\lambda-1}{k} x(j+k,n+\lambda-1) = \sum_{k=0}^{\lambda-1} x(j+k,n+\lambda-1) \]

Then, since $\lambda\geq 3$ and $\binom{\lambda-1}{0} = \binom{\lambda-1}{\lambda-1} = 1$,
\[ \sum_{k=1}^{\lambda-2} \binom{\lambda-1}{k} x(j+k,n+\lambda-1) = \sum_{k=1}^{\lambda-2} x(j+k,n+\lambda-1) \]
from which we can derive that $x(j+k, n+\lambda -1) = 0$ for all $1\leq k\leq \lambda - 2$ since $\binom{\lambda-1}{k} > 1$ for all $1\leq k\leq \lambda - 2$. This is a contradiction with SIR.
\end{proof}

Following Lemma~\ref{lemma:group}, we can conclude that under the requirement of SIR, being both Sybil-proof and collusion-proof is impossible for a path mechanism.

\begin{proposition}\label{prop:spcp}
An SIR path mechanism cannot be both Sybil-proof and collusion-proof.
\end{proposition}

%\begin{proof}
%    Note that the collusion-proofness can directly derive the property of $\lambda$-collusion proofness for any $\lambda>1$. Then with Lemma~\ref{lemma:group} we can conclude this proposition.
%\end{proof}

However, we will show that the property of 2-collusion-proof, which says that two agents cannot collude to make a gain, can be satisfied with Sybil-proofness in Section~\ref{sec:dgm}.

The impossibility assumes strongly IR. However, even if we weaken the condition to be IR, Theorem~\ref{thm_tm} shows that the only mechanisms to satisfy both SP and CP are limited to very special mechanisms called two-headed mechanism. % Another consideration to avoid this conflict is weaken the property of SIR to normal IR. However, the class of IR, SP and CP mechanisms is still very limited.

\begin{mechanism}
A path mechanism is a \textbf{two-headed mechanism} if its reward function satisfies %$x(1,1)=a+b$,
% \[ x(j,n) = \begin{cases}
% a & \text{if } j = 1\\
% b & \text{if } j = n
% \end{cases} \quad\text{for all } n>1 \]
\[ x(j,n) = \begin{cases}
a + b & \text{if } j = 1 \text{ and } n = 1 \\
a & \text{if } j = 1 \text{ and } n > 1 \\
b & \text{if } j = n \text{ and } n > 1 \\
0 & \text{otherwise}
\end{cases} \]
% and $x(j,n) = 0$ otherwise,
where $a,b\geq0$ are constants.
\end{mechanism}

Intuitively, a two-headed mechanism only allocates positive rewards to the first agent and the winner on the winning path and all the other agents receive a zero reward.

\begin{theorem}
\label{thm_tm}
A path mechanism is IR, SP and CP if and only if it is a two-headed mechanism.
\end{theorem}

\begin{proof}
    (``$\Leftarrow$") First, it is easy to show that a two-headed mechanism is IR, SP and CP since $x(j,n)\geq 0$ and
    \[\sum_{j=1}^n x(j,n) = a + b\]
    for all $j,n\in\mathbb{Z}^+$ and $j \leq n$.
    
    (``$\Rightarrow$") Then we show that these three properties determine a two-headed mechanism. From the definition of CP, we know that the mechanism is $\lambda$-CP for all $\lambda>1$. Then from Lemma~\ref{lemma:group} we know that a mechanism with SP and $\lambda$-CP satisfies $x(j,n) = 0$ for all $1<j<n$ and $n>1$. Hence, from SP and 2-CP we can derive that
    \[ x(1,2) = x(1,3) + x(2,3) = x(1,3) = \dots = x(1,n) \]
    and
    \[ x(2,2) = x(2,3) + x(3,3) = x(3,3) = \dots = x(n,n) \]
    for all $n>1$.
    
    Hence, there exist two constants $a\geq0$ and $b\geq0$ such that $x(1,n)= a$, $x(n,n)=b$ for all $n>1$ and $x(1,1) = x(1,2) + x(2,2) = a+b$.
\end{proof}

\section{Double Geometric Mechanism}\label{sec:dgm}
In this section, we characterize a class of reward mechanisms with the desirable properties. The mechanism we will characterize is called \textbf{Double Geometric Mechanism (DGM)}, which is a path mechanism defined by two parameters.

\begin{mechanism}
A path mechanism is an $(\alpha, \delta)$ double geometric mechanism ($(\alpha, \delta)$-DGM), for $0<\alpha<1$ and $\delta > 0$, if its reward function satisfies
\[ x(j,n) = (1-\alpha)^{j-1}\alpha^{n-j}\delta \]
for all $j,n\in \mathbb{Z}^+$ such that $j\leq n$.
\end{mechanism}

Intuitively, $(\alpha, \delta)$-DGM has two fractions $\alpha^{n-j}$ and $(1-\alpha)^{j-1}$, which are controlled by the distance to the winner and the requester respectively. Note that if $\alpha<1/2$, the reward is strictly monotone decreasing with the depth on the winning path, while if $\alpha>1/2$, the reward is strictly monotone increasing with the depth on the winning path. Theorem~\ref{thm:dgm} shows that $(\alpha, \delta)$-DGM can satisfy all the desirable properties defined in our model.

\begin{theorem}\label{thm:dgm}
  If a path mechanism is an $(\alpha, \delta)$-DGM with $\frac{\rho}{1+\rho} \leq \alpha < \frac{1}{2}$ $(0<\rho<1)$, then it is IC, SIR, BC, SP, 2-CP and $\rho$-SS.
\end{theorem}

\begin{proof}
Suppose the winning path is $(i_1, i_2, \dots, i_n)$ of length $n$ when the agents behave truthfully. For an agent $i_j$ with $1\leq j<n$, $x_{i_j} = x(j,n)$, if she did not query all her neighbours, then she will be either on the winning path of length $\geq n$ or not on the winning path. By doing so, her reward is either $x(j,n+k) = (1-\alpha)^{j-1}\alpha^{n-j+k}\delta \leq (1-\alpha)^{j-1}\alpha^{n-j}\delta = x(j,n)$ for some $k\geq 0$, or zero, which is not better than behaving truthfully. For the winner $i_n = w$, if she did not provide the answer and further queried her children, then she would be either on the winning path of length $> n$ or not on the winning path, which gives her a reward either $x(n,n+k) = (1-\alpha)^{n-1}\alpha^{k}\delta < (1-\alpha)^{n-1}\delta = x(n,n)$ for some $k\geq 1$ or zero. Hence, $(\alpha, \delta)$-DGM is IC.
    
    For all $j,n\in \mathbb{Z}^+$, $j\leq n$, we have $x(j,n) = (1-\alpha)^{j-1}\alpha^{n-j}\delta > 0$. Hence, $(\alpha, \delta)$-DGM is SIR.
    
    For all $j,n\in \mathbb{Z}^+$, $j\leq n$ and $0<\alpha<\frac{1}{2}$, we have
    \begin{align*}
        \sum_{j=1}^n x(j,n) & = \sum_{j=1}^n (1-\alpha)^{j-1}\alpha^{n-j}\delta \\
        & = \frac{\alpha^n\delta}{1-\alpha} \sum_{j=1}^n \left( \frac{1-\alpha}{\alpha} \right)^j \\
        & = \frac{\alpha^n - (1-\alpha)^n}{2\alpha - 1}\cdot \delta \leq \delta
    \end{align*}
    Hence, $(\alpha, \delta)$-DGM is BC.
    
    For all $j,n\in \mathbb{Z}^+$, $j< n$, we have
    \begin{align*}
       x(j,n) =\ & \alpha x(j,n) + (1-\alpha) x(j,n) \\
        =\ & \alpha (1-\alpha)^{j-1}\alpha^{n-j}\delta + (1-\alpha) (1-\alpha)^{j-1}\alpha^{n-j}\delta \\
        =\ & (1-\alpha)^{j-1}\alpha^{n+1-j}\delta + (1-\alpha)^{j}\alpha^{n+1-(j-1)}\delta \\
        =\ & x(j,n+1) + x(j+1,n+1)
    \end{align*}
    Hence, $(\alpha, \delta)$-DGM is 2-CP. According to Proposition~\ref{prop:sp}, it is also SP.
    
    Finally, for all $j,n\in \mathbb{Z}^+$, $j< n$, we have
    \[ \frac{x(j,n)}{x(j+1,n)} = \frac{\alpha}{1-\alpha} \geq \frac{\rho/(1+\rho)}{1-\rho/(1+\rho)} = \rho \]
    Hence, $(\alpha, \delta)$-DGM is $\rho$-SS.
\end{proof}

$(\alpha,\delta)$-DGM satisfies all the desirable properties. Then we wonder are they the only mechanisms to satisfy these properties. Under some mild conditions, we will prove that $(\alpha,\delta)$-DGM is indeed the only mechanism to satisfy all the properties.

Note that if there are only two agents on the winning path, it is a very special case that has almost no constraints (the reward can be assigned arbitrarily), which suggests that it acts like an initial condition to the reward function. To satisfy the property of $\rho$-SS, we should have $x(1,2)\geq \rho x(2,2)$ and we let $x(1,2)=\rho x(2,2)$, which is the simplest way to construct the initial condition. We say a mechanism uses $\rho$-base condition if its reward function $x$ satisfies $x(1,2) = \rho x(2,2)$.
%On the other hand, we consider to characterize the $(\alpha,\delta)$-DGM with desired properties. Naturally, we start with an initial case where there is only one agent on the winning path. Suppose we give her a reward of $\delta$, i.e., $x(1,1)=\delta$. Then when there are two agents on the winning path, according to the property of $\rho$-SS, we should make $x(2,1)\geq \rho x(2,2)$ and we let $x(2,1)=\rho x(2,2)$, which is the simplest way to construct the initial case. We call such an initial case the $(\delta,\rho)$-base case. Since the $(\delta,\rho)$-base case only gives an initial case of the reward function, it will not lose the generality of the result. 
Theorem~\ref{thm_cha} proves that under the base condition, $(\alpha, \delta)$-DGM is the only kind of mechanism to satisfy all the properties.
%Now we show that all desired properties along with this base case can completely characterize an $(\alpha, \delta)$-DGM.

\begin{theorem}
\label{thm_cha}
If a path mechanism is IC, SIR, BC, SP, 2-CP, $\rho$-SS and uses $\rho$-base condition, then it is an $(\alpha, \delta)$-DGM with $\alpha = \frac{\rho}{1+\rho}$.
\end{theorem}

\begin{proof}
    First consider the value $x(j,n)$, $x(j+1,n)$, $x(j,n+1)$, $x(j+1,n+1)$ and $x(j+2,n+1)$ for some $j,n\in \mathbb{Z}^+$ and $j<n$. Suppose that $x(j,n) = \gamma x(j+1,n)$, $x(j,n+1) = \gamma_2 x(j+1,n+1)$ and $x(j+1,n+1) = \gamma_1 x(j+2,n+1)$.
    
    Then according to the property of SP and 2-CP, we know that $x(j,n) = x(j,n+1) + x(j+1,n+1)$ and $x(j+1,n) = x(j+1,n+1) + x(j+2,n+1)$, from which we have
    \begin{align*}
        x(j,n) &= (1+\gamma_2)\gamma_1 x(j+2,n+1) \\
        x(j+1,n) &= (1+\gamma_1)x(j+2,n+1)
    \end{align*}
    
    Hence, by $x(j,n) = \gamma x(j+1,n)$, we have
%   \[ \gamma(1+\gamma_1) = \gamma_1(1+\gamma_2) \]
%    Or alternatively,
    \[ \gamma_2 = \gamma\left( 1+\frac{1}{\gamma_1} \right) - 1 \]
    where according to the property of $\rho$-SS, we have $\gamma\geq\rho$, $\gamma_1\geq\rho$ and $\gamma_2\geq\rho$.
    
    If $x(j,n)=\rho x(j+1,n)$, i.e., $\gamma = \rho$, then
    \[ \gamma_2 = \rho\left( 1+\frac{1}{\gamma_1} \right) - 1 \geq \rho \]
    which suggests that $\gamma_1\leq \rho$. Hence $\gamma_1 = \rho$ and then $\gamma_2 = \rho$. Note that we have $x(1,2) = \rho x(2,2)$ as the base condition, so from above we have $x(1,3) = \rho x(2,3)$ and $x(2,3) = \rho x(3,3)$. Then, by induction, $x(j,n) = \rho x(j+1,n)$ holds for any $j,n\in \mathbb{Z}^+$ and $j<n$. Therefore, the following recursive relation holds for any $j,n\in \mathbb{Z}^+$ and $j\leq n$:
    \[ \begin{cases}
    x(j,n+1) = \frac{\rho}{1+\rho}x(j,n) \\
    x(j+1,n+1) = \frac{1}{1+\rho}x(j,n)
    \end{cases}\]
    
    Denote $\delta = x(1,1)>0$, then we can derive that
    \[ x(j,n) = \left(\frac{1}{1+\rho}\right)^{j-1}\left(\frac{\rho}{1+\rho}\right)^{n-j}\delta = (1-\alpha)^{j-1}\alpha^{n-j}\delta \]
    with $\alpha = \rho/(1+\rho)$ for any $j,n\in \mathbb{Z}^+$ and $j\leq n$, which is an $(\alpha, \delta)$-DGM. Also according to Theorem~\ref{thm:dgm}, the property of IC, SIR and BC will not be hurt. Therefore, these properties will uniquely determine the $(\alpha, \delta)$-DGM.
\end{proof}

So far, we have shown that $(\alpha,\delta)$-DGM is the only kind of mechanism to satisfy the properties under the $\rho$-base condition. Note that, without the $\rho$-base condition, we may get a different mechanism to satisfy all the properties, but it is still a DGM-like mechanism with a bounded difference. It suggests that the space of mechanisms removed by this base condition is also limited. Therefore, the base condition does not significantly hurt the generality of our result.
%From the discussion above, we know that under the dominant strategy implementation, the $(\alpha,\delta)$-DGM is the first and the only mechanism we can find with the desired properties in our model.

As Proposition~\ref{prop:spcp} shows that SP and CP cannot be held together under SIR, the above mechanism can only satisfy 2-CP. In the following, we show how much extra gain a group of agents could get if they collude together and a weaker concept of CP. Furthermore, we will investigate the cost of the requester and show how to minimize it.

%Following we discuss two other important issues. The first one is the approximation of collusion-proofness. Although we have showed that the property of 2-CP as a localized collusion-proofness can be achieved, we also consider an upper bound of the extra reward if one can collude with a group with unlimited capability. The other one is cost minimization where the requester want to minimize the expense and we show an alternative characterization of $(\alpha,\delta)$-DGM without $(\delta,\rho)$-base case.

\subsection{Approximation of Collusion-proofness}
\begin{definition}
A path mechanism is called \textbf{$\beta$-approximate collusion-proof ($\beta$-ACP)} if its reward function $x$ satisfies
\[ x(j,n) \leq \beta \sum_{k=0}^m x(j+k,n+m) \]
for all $j,n,m\in \mathbb{Z}^+$, $j\leq n$.
\end{definition}

Intuitively, the property of $\beta$-approximate collusion-proof ensures that if some agents pretend to be a single agent, they can achieve at most $\beta$ times of their original reward. However, we show that a constant approximation cannot be achieved along with the other properties.

\begin{proposition}
An IC, SIR, BC, SP, 2-CP and $\rho$-SS path mechanism with $\rho$-base condition cannot be $(1+\epsilon)$-ACP for any $\epsilon>0$.
\end{proposition}

\begin{proof}
    Note that an IC, SIR, BC, SP, 2-CP and $\rho$-SS path mechanism with $\rho$-base condition must be an $(\alpha,\delta)$-DGM with $\alpha = \rho /(1+\rho)$. Then $(1+\epsilon)$-ACP implies that
    \[ (1-\alpha)^{j-1}\alpha^{n-j}\delta \leq (1+\epsilon) \sum_{k=0}^m (1-\alpha)^{j+k-1}\alpha^{n+m-j-k}\delta \]
    from which we can derive that
    \[ (1-\alpha)^{m+1} - \alpha^{m+1} \geq \frac{1-2\alpha}{1+\epsilon} \]
    
    However, for any given $\epsilon>0$ and $0<\alpha<1/2$, there exists $N>0$ such that for any $m>N$, $(1-\alpha)^{m+1} - \alpha^{m+1} < (1-2\alpha)/(1+\epsilon)$ since the left hand side approaches to 0 when $m$ approaches to $\infty$. Hence, it cannot be $(1+\epsilon)$-ACP.
\end{proof}

On the other hand, an exponential approximation is easy to be achieved by an $(\alpha,\delta)$-DGM.

\begin{theorem}
If a path mechanism is an $(\alpha,\delta)$-DGM with $0<\alpha<\frac{1}{2}$, then it is $2^m$-ACP, where $m$ is the number of agents who collude together.
\end{theorem}

\begin{proof}
    If an $(\alpha,\delta)$-DGM is $2^m$-ACP, it suggests that
    \[ (1-\alpha)^{j-1}\alpha^{n-j}\delta \leq 2^m \sum_{k=0}^m (1-\alpha)^{j+k-1}\alpha^{n+m-j-k}\delta \]
    from which we can derive the equivalent inequality that
    \[ (1-\alpha)^{m+1} - \alpha^{m+1} \geq \frac{1-2\alpha}{2^m} \]
    or being rearranged as
    \[ (1-\alpha)^{m+1} - \frac{1-\alpha}{2^m} \geq \alpha^{m+1} - \frac{\alpha}{2^m} \]
    
    Notice that for function $f(y) = y\left( y^m - \frac{1}{2^m} \right)$, $f(y)<0$ for any $0<y<1/2$ and $f(y)>0$ for any $1/2<y<1$. Since $0<\alpha<1/2$, the above inequality always holds.
\end{proof}

% \section{Cost Minimization}\label{sec:dis}
% In this section, we discuss several other important issues including cost minimization and collusion-proofness.

\subsection{Cost Minimization}
Next, we consider the case where the requester is willing to minimize her cost to query the answer.

\begin{definition}
A path mechanism is \textbf{of minimum cost} over a class of path mechanisms $\mathcal{X}$ if its reward function $x$ satisfies
\[ x\in \arg\min_{x\in \mathcal{X}} \sum_{j=1}^n x(j,n)  \]
for all $n\in \mathbb{Z}^+$.
\end{definition}

Intuitively, a path mechanism is of minimum cost if it can minimize the total reward distributed to the agents on the winning path. Notice that when we minimize the cost, we consider time-critical mechanisms, which are essential to the real-world application~\cite{pickard2011time}. A path mechanism is time-critical if the winner always takes the maximum reward, i.e., $x(n,n) = \max_j x(j,n)$ for any $n\geq 1$. Now we show that the $(\alpha, \delta)$-DGM also characterizes the space of time-critical mechanisms of minimum cost.

\begin{theorem}
A path mechanism is of minimum cost over a class of time-critical path mechanisms that satisfy IC, SIR, BC, SP, 2-CP, $\rho$-SS and have $x(1,1) = \delta$ if and only if it is an $(\alpha, \delta)$-DGM with $\alpha = 1/2$.
\end{theorem}

\begin{proof}
    First we show a lower bound of the total cost before we prove the statement. According to the property of SP and 2-CP, denote $\sum_{j=1}^n x(j,n) = R_n$, then we have
    \[ 2R_{n+1} = R_n + (x(1,n+1) + x(n+1,n+1)) \]
    with $R_1 = R_2 = \delta$. Hence, to minimize the total cost is to minimize the term $x(1,n+1) + x(n+1,n+1)$. Suppose that $x(1,2) = \gamma x(2,2)$, $x(1,3) = \gamma_2 x(2,3)$ and $x(2,3) = \gamma_1 x(3,3)$ where $\rho\leq \gamma,\gamma_1,\gamma_2\leq 1$ with $\gamma = \rho + \epsilon$. Then similar to the process in Theorem~\ref{thm_cha}, we can derive that $\gamma_2 = \gamma(1+1/\gamma_1) - 1$ and $\gamma_1 \leq (\rho+\epsilon)/(1-\epsilon)$. Hence,
    \[ x(1,3) + x(3,3) \geq \left(1-\frac{2(\rho + \epsilon)}{(\rho+1)(1+\rho+\epsilon)}\right)\delta \]
    
    As we want to minimize the cost, the minimum value will be reached if $\gamma = \gamma_1 = \gamma_2 = 1$, which means $x(1,2)=x(2,2)$ and $x(1,3)=x(2,3)=x(3,3)$. By induction, if for some $n$, $x(j,n) = x(j+1,n)$ for any $j<n$, suppose $x(j,n+1)=\gamma_j x(j+1,n+1)$ for any $j\leq n$, then we have $1+\gamma_{n}=\gamma_{n}(1+\gamma_{n-1}) = \dots = \gamma_n\cdots\gamma_2(1+\gamma_1)$ with $\rho\leq \gamma_j\leq 1$. So that the solution could only be $\gamma_1 = \dots = \gamma_n = 1$. Therefore, $x(j,n)=x(j+1,n)$ for any $j,n\in\mathbb{Z}^+$ and $j<n$.
    
    (``$\Leftarrow$") First, an $(\alpha,\delta)$-DGM is IC, SIR, BC, SP, 2-CP, $\rho$-SS and have $x(1,1) = \delta$ by Theorem~\ref{thm:dgm}. Then, for an $(\alpha,\delta)$-DGM with $\alpha = 1/2$, the total reward distributed is
    \[ \sum_{j=1}^n x(j,n) = \sum_{j=1}^n (1-\alpha)^{j-1}\alpha^{n-j}\delta = n\left(\frac{1}{2}\right)^{n-1} \delta \]
    for any $n\geq 1$. Besides, from the above we know $2R_{n} \geq R_{n-1} + \delta/2^{n-2}$ for $n\geq 3$ and $R_2 = \delta$. Then $R_n \geq (n\delta)/2^{n-1} $. Since $(\alpha,\delta)$-DGM with $\alpha = 1/2$ meets this lower bound, then it is the mechanism of minimum cost.

    (``$\Rightarrow$") By the equation of $x(j,n) = x(j+1,n)$ and $x(1,1)=\delta$, we can derive that $x(j,n) = \delta/2^{n-1}$, which is the same as for the $(1/2, \delta)$-DGM. Therefore, the mechanism of minimum cost over these properties uniquely determines the $(1/2, \delta)$-DGM.
\end{proof}

\section{Conclusion}\label{sec:con}
We have investigated Sybil-proof answer querying mechanisms on networks in a dominant strategy implementation. We proposed a class of double geometric mechanisms (DGM) to against Sybil attacks and characterized its uniqueness under other important properties such as IC, IR and 2-CP. We also characterized the mechanisms for minimizing the requester's reward expenses and illustrated the performance of the mechanisms in terms of the approximation of collusion-proofness.
There are several other interesting aspects worth further investigation. There is a gap between constant approximation and exponential approximation of collusion-proofness. The characterization of Sybil-proof mechanisms in a general directed graph is also missing.

%\section*{Acknowledgments}

%\appendix

%\section{\LaTeX{} and Word Style Files}\label{stylefiles}

%% The file named.bst is a bibliography style file for BibTeX 0.99c
\bibliographystyle{named}
\bibliography{ijcai20}

\end{document}